\documentclass[11pt]{article}

\usepackage[a4paper,margin=1in]{geometry}
\usepackage{times}
\usepackage{mathptmx}
\usepackage{microtype}
\usepackage[utf8]{inputenc}
\usepackage{amsmath}
\usepackage{booktabs}
\usepackage{algorithm}
\usepackage{algorithmic}
\usepackage{tikz}
\usepackage{hyperref}
\usepackage{graphicx}
\usepackage{float}
\usepackage{multirow}
\usepackage{booktabs}
\usepackage{makecell}
\usepackage{placeins}
\usepackage{amsthm}
\newtheorem{theorem}{Theorem}[section]

\usetikzlibrary{positioning,arrows.meta,shapes.multipart}

\title{Constraint-Aware Generative Re-ranking for Multi-Objective Optimization in Advertising Feeds}

\author{
Chenfei Li, Hantao Zhao, Weixi Yao, Ruiming Huang, Rongrong Lu, Geng Tian, Dongying Kong \\
Bilibili Inc \\
\{lichenfei, zhaohantao, yaoweixi, huangruiming, lurongrong, tiangeng, limin07\}@bilibili.com
}

\date{}

\begin{document}

\maketitle

\begin{abstract}
Optimizing reranking in advertising feeds is a constrained combinatorial problem, requiring simultaneous maximization of platform revenue and preservation of user experience.
Recent generative ranking methods enable listwise optimization via autoregressive decoding, but their deployment is hindered by high inference latency and limited constraint handling.

We propose a constraint-aware generative reranking framework that transforms constrained optimization into bounded neural decoding. Unlike prior approaches that separate generator and evaluator models, our framework unifies sequence generation and reward estimation into a single network.

We further introduce constraint-aware reward pruning, integrating constraint satisfaction directly into decoding to efficiently generate optimal sequences.
Experiments on large-scale industrial feeds and online A/B tests show that our method improves revenue and user engagement while meeting strict latency requirements, providing an efficient neural solution for constrained listwise optimization.
\end{abstract}

\section{Introduction}

In modern feed-based recommendation systems, the final stage is commonly referred to as the \emph{reranking} stage, where a mixed list of organic content (recommendations) and sponsored content (advertisements) is presented to users.
A common industrial practice dynamically inserts advertisements into a pre-ranked organic feed.
While organic ranking aims to maximize user engagement and long-term experience, ad ranking optimizes expected platform revenue.
Consequently, reranking becomes a full-page multi-objective maximization problem~\cite{ads_allocation}.

Due to business and policy requirements, ad insertion must satisfy strict constraints, including position constraints (e.g., preventing overly dense ad placements) and load constraints (e.g., limiting ads to at most $x\%$ of the page).
Prior work addresses adaptive ad exposure under hierarchical constraints using dynamic knapsack-style optimization~\cite{hierarchical_constraint}.

With the widespread adoption of deep learning recommendation models (DLRMs), various reranking approaches have been proposed, such as Personalized Re-ranking Models (PRM)~\cite{prm}.
These methods leverage contextual signals to improve ranking accuracy.
More advanced approaches optimize listwise objectives by explicitly modeling inter-item dependencies and capturing real-time user preferences.
In industrial systems, reranking is often implemented as a two-stage Generator--Evaluator (GE) framework:
a generator proposes candidate sequences, and an evaluator estimates list-level utility for final selection.
Although effective, this separation increases inference overhead and complicates constraint enforcement.

Recently, generative recommenders~\cite{asl} have reframed reranking as a sequence generation problem.
Autoregressive listwise optimization has been explored for multi-objective optimization~\cite{taobao_gr},
while non-autoregressive models aim to reduce decoding latency~\cite{nar}.
Flow-based generative approaches further enable global sequence modeling~\cite{gfn}.
However, in heterogeneous feeds containing both organic and sponsored content, reranking must simultaneously address multi-objective optimization, inference latency, and constraint-aware ad allocation.

Motivated by these challenges, we propose a \textbf{Constraint-Aware Generative Re-ranking (CGR)} framework,
which transforms constrained combinatorial optimization into bounded autoregressive decoding solvable within a single neural inference pass.
To the best of our knowledge, this is the first work that integrates policy-level constraints directly into autoregressive decoding to obtain optimal constrained sequences in generative reranking.

CGR consists of two key components.
First, a variable-length listwise autoregressive training network models multi-objective user intent through attention mechanisms and multi-channel mixture-of-experts (MoE) modules, enabling context-aware and interest-aware representation learning.
Second, during inference, a constraint-aware reward network generates and evaluates candidate sequences.
Inspired by masked tensor propagation strategies similar to M-FALCON~\cite{asl}, we amortize decoding cost through structured masking and further improve efficiency via reward-based pruning, enabling fast and optimal constrained sequence generation.

\FloatBarrier

\section{Related Work}

\subsection{Reranking and Listwise Optimization}

Reranking serves as the final stage of industrial recommendation systems, where a ranked list is refined to optimize downstream objectives.
Early approaches focus on pointwise or pairwise ranking models, while more advanced methods adopt listwise optimization to explicitly model inter-item dependencies.

Personalized Re-ranking Models (PRM)~\cite{prm} incorporate contextual features to improve ranking accuracy.
Subsequent work such as DLCM and Seq2Slate leverages recurrent or attention-based architectures to capture list-level interactions.
These methods primarily optimize user engagement metrics but do not explicitly address heterogeneous content allocation or business constraints.

In advertising feeds, reranking becomes a multi-objective optimization problem balancing user experience and revenue.
\cite{ads_allocation} formulates feed advertising as a constrained optimization problem.
\cite{hierarchical_constraint} further introduces hierarchical exposure control under practical business constraints.
However, these approaches typically rely on explicit combinatorial solvers rather than neural sequence generation.

\subsection{Generative Recommendation and Autoregressive Ranking}

Recently, generative recommenders have reframed ranking as a sequence generation task.
\cite{asl} demonstrates the scalability of large sequential transducers for generative recommendation.
Autoregressive reranking has been applied to multi-objective optimization in industrial settings~\cite{taobao_gr}.
Generative Flow Networks (GFN)~\cite{gfn} model the distribution over ranked lists through flow-based objectives.
Non-autoregressive generative models~\cite{nar} aim to reduce decoding latency by parallel prediction.

Despite their success, most generative ranking methods suffer from two limitations:
(1) decoding latency grows with sequence length, and
(2) constraint handling is typically external to the decoding process.
Business rules such as advertisement load constraints are often enforced through post-processing heuristics.

\subsection{Constrained Allocation in Advertising Systems}

Advertising allocation in feeds has been extensively studied under constrained optimization frameworks.
Dynamic programming and knapsack-style solvers are widely adopted to satisfy load and spacing constraints~\cite{hierarchical_constraint, ads_allocation}.
These approaches provide theoretical guarantees but require explicit enumeration or approximate solvers, which complicates integration with deep neural ranking models.

In practice, many industrial systems adopt a Generator--Evaluator (GE) paradigm,
where a generator proposes candidate sequences and an evaluator estimates list-level utility.
Although effective, this two-stage separation increases inference overhead and makes constraint enforcement less efficient.

\subsection{Our Contribution}

Our work bridges generative ranking and constrained allocation.
We transform constrained combinatorial optimization into bounded autoregressive decoding.
Unlike prior GE-based systems, we unify sequence generation and reward estimation into a single network.
Furthermore, constraint-aware reward pruning integrates constraint satisfaction directly into decoding, providing an efficient neural solution with theoretical optimality guarantees under bounded structural assumptions.

To our knowledge, this is the first work that embeds policy-level advertisement constraints directly into autoregressive generative reranking to obtain globally optimal constrained sequences in production-scale advertising feeds.

\FloatBarrier

\section{Problem Formulation}

We formulate advertising feed reranking as a constrained combinatorial optimization problem.

\subsection{Candidate Set}

Let
\begin{equation}
\mathcal{X} = \{x_1, x_2, \dots, x_N\}
\end{equation}

denote the candidate pool for a given user request, where items consist of:

\begin{itemize}
\item Organic content (natural recommendations)
\item Sponsored advertisements
\end{itemize}

Each item $x_i$ is associated with feature representation
including user context, item features, and positional encoding.

\subsection{Ranking Sequence}

A ranking sequence is defined as an ordered list:

\begin{equation}
A = (a_1, a_2, \dots, a_L),
\end{equation}

where $L$ is the page size and $a_j \in \mathcal{X}$.

Let $\mathcal{S}$ denote the space of all possible permutations:

\begin{equation}
\mathcal{S} = \text{Perm}(\mathcal{X}).
\end{equation}

In the unconstrained case, $|\mathcal{S}| = N!$,
which is intractable for real-time inference.

\subsection{List-Level Reward}

The objective of reranking is to maximize a list-level reward function:

\begin{equation}
R(A) = f_\theta(A),
\end{equation}

where $f_\theta(\cdot)$ is a neural model estimating
multi-objective utility including:

\begin{itemize}
\item Advertising revenue
\item User engagement
\item Experience penalty
\end{itemize}

Concretely, the reward is defined as:

\begin{equation}
R(A) = \sum_{i=1}^{L} \left( V_i + N_i - P_i \right),
\end{equation}

where $V_i$ denotes monetization value,
$N_i$ denotes engagement benefit,
and $P_i$ denotes advertisement exposure penalty.

\subsection{Business Constraints}

Industrial advertising feeds impose strict structural constraints.

\paragraph{(1) Advertisement Load Constraint}

Let $\#Ad(A)$ denote the number of ads in sequence $A$.
We require:

\begin{equation}
\#Ad(A) \le K,
\end{equation}

where $K$ is a small constant (e.g., $K \le 2$).

\paragraph{(2) Position Spacing Constraint}

Let $pos(a)$ denote the position of advertisement $a$.
For any two ads $a_i, a_j$:

\begin{equation}
|pos(a_i) - pos(a_j)| \ge \Delta,
\end{equation}

ensuring minimum spacing between ads.

\paragraph{(3) Structural Rules}

Additional rules include:

\begin{itemize}
\item Valid insertion intervals
\item Large-ad placement constraints
\item User-level exposure frequency control
\end{itemize}

\subsection{Constrained Optimization Objective}

Let $\mathcal{F} \subseteq \mathcal{S}$ denote the feasible set satisfying all constraints:

\begin{equation}
\mathcal{F} = \{ A \in \mathcal{S} \mid C(A) \le 0 \}.
\end{equation}

The reranking problem becomes:

\begin{equation}
\max_{A \in \mathcal{F}} R(A).
\end{equation}

\subsection{Key Observation}

Although $|\mathcal{S}|$ grows factorially,
the structural constraint $\#Ad(A) \le K$ bounds the feasible set $\mathcal{F}$ to a significantly smaller combinatorial space.

This bounded structure enables efficient decoding strategies
that preserve optimality within $\mathcal{F}$,
which forms the foundation of our Constraint-Aware Generative Re-ranking framework.

\FloatBarrier

\section{Generative Ranking Model}

We implement CGR as a unified listwise autoregressive model that jointly predicts exposure, click, and list-level reward within a single forward pass, as illustrated in Figure~\ref{fig:model_struct}

\begin{figure}[H]
    \centering
    \includegraphics[width=0.8\textwidth]{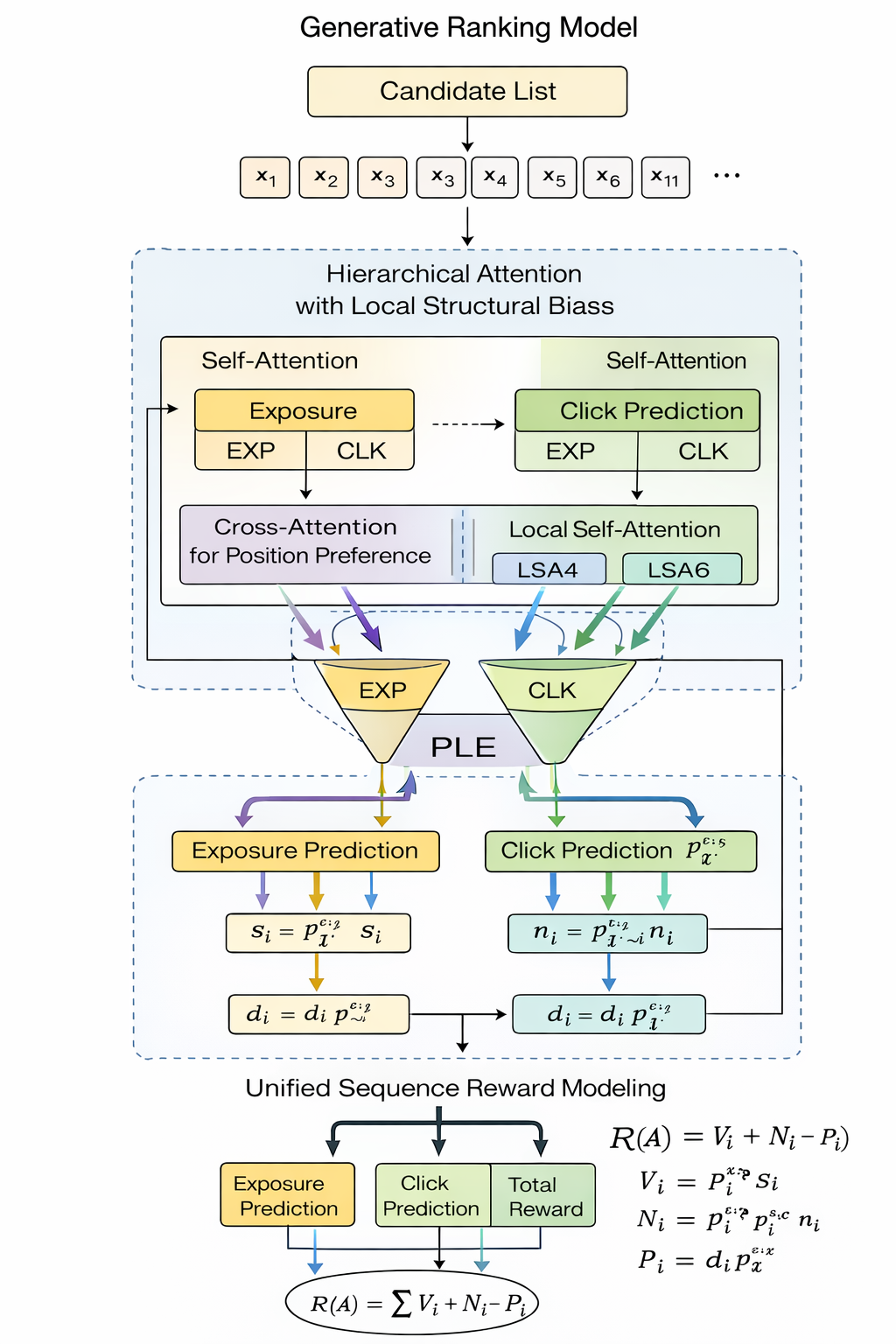}
    \caption{
    Model struct of CGR.
    }
    \label{fig:model_struct}
\end{figure}

The training objective is a weighted multi-task loss:

\[
\mathcal{L} = \lambda_{exp} \mathcal{L}_{exp} + \lambda_{clk} \mathcal{L}_{clk},
\]

where each component is optimized using binary cross-entropy.

\subsection{Multi-Task Listwise Architecture}

Given a candidate list of maximum length $L=11$, each item is represented by:

\begin{equation}
x_i = [u_i, c_i, v_i, p_i],
\end{equation}

where $u_i$ denotes user features, $c_i$ contextual features,
$v_i$ item features, and $p_i$ positional embeddings.

The model jointly optimizes two primary tasks:

\begin{itemize}
\item Exposure prediction (EXP)
\item Click prediction (CLK)
\end{itemize}

Unlike traditional reranking systems that separate generator and evaluator networks,
our model directly produces both item-level predictions and sequence-level reward,
thus unifying generation and evaluation.

\subsection{Hierarchical Attention with Local Structural Bias}

To model inter-item dependency, we employ a hierarchical attention design:

\paragraph{(1) Self-Attention Experts}

Task-specific and shared self-attention heads capture global sequence interactions:

\begin{equation}
H^{self} = \text{Attention}(X, \text{mask}),
\end{equation}

where causal masks enforce autoregressive structure.

\paragraph{(2) Cross-Attention for Position Preference}

We explicitly model user–item positional interaction via cross-attention:

\begin{equation}
H^{pos} = \text{CrossAttention}(U_{side}, I_{side}),
\end{equation}

where user-side and item-side encoders produce position-aware embeddings.
This mechanism captures structural bias such as first-position preference and ad-slot sensitivity.

\paragraph{(3) Local Self-Attention (LSA)}

To encode locality in feed browsing behavior,
we introduce Local Self-Attention (LSA) with fixed windows (size 4 and 6):

\begin{equation}
H^{lsa_k} = \text{Attention}(X, M_k),
\end{equation}

where $M_k$ is a predefined band mask restricting attention to local neighborhoods.

This reduces over-smoothing from global attention
and encodes practical feed exposure patterns.

\subsection{Multi-Expert Fusion via PLE}

All attention outputs are fused using a multi-task Progressive Layered Extraction (PLE) gate:

\begin{equation}
H^{exp} = \text{PLE}_{exp}(H^{shared}, H^{exp}, H^{clk}),
\end{equation}

\begin{equation}
H^{clk} = \text{PLE}_{clk}(H^{shared}, H^{clk}).
\end{equation}

The EXP branch integrates 12 experts (shared, exp, clk),
while the CLK branch integrates 8 experts (shared, clk),
enabling cross-task knowledge transfer while preserving task specificity.

\subsection{Unified Sequence Reward Modeling}

The model outputs:

\begin{equation}
\hat{p}^{exp}_i, \quad \hat{p}^{clk}_i,
\end{equation}

and computes list-level reward directly:

\begin{equation}
R(A) = \sum_i V_i + N_i - P_i,
\end{equation}

where

\begin{align}
V_i &= 
\begin{cases}
\hat{p}^{clk}_i \hat{p}^{exp}_i s_i, & \text{if CPA} \\
\hat{p}^{exp}_i s_i, & \text{otherwise}
\end{cases} \\
N_i &= \hat{p}^{exp}_i \hat{p}^{clk}_i n_i \\
P_i &= d_i \hat{p}^{exp}_i
\end{align}

represent monetization value, user experience reward, and advertisement penalty respectively.

This design eliminates the need for a separate evaluator network,
as reward estimation is integrated into generation.

\FloatBarrier

\section{Bounded Decoding}

Direct autoregressive permutation over $N$ candidates yields factorial complexity.
However, advertising feeds impose strict structural constraints,
notably a bounded advertisement count per list.

\subsection{Structural Constraint Bounding}

Let $K$ denote the maximum number of advertisements allowed per list (in production $K \le 2$).
Instead of generating full permutations,
we restrict decoding to feasible insertion operations.

The feasible search space becomes:

\begin{equation}
\mathcal{F} = \{ A \mid \#Ad(A) \le K, \text{position constraints satisfied} \}.
\end{equation}

Thus decoding reduces from:

\begin{equation}
O(N!)
\end{equation}

to bounded enumeration over ad insertion patterns:

\begin{equation}
O(K \cdot L).
\end{equation}

\subsection{Two-Stage Inference}

Inference is implemented via a structured two-stage procedure:

\paragraph{Stage I: Constrained Single-Ad Insertion}

All feasible single-ad insertions within position bounds
($min\_pos$, $max\_pos$) are constructed.
The reward model evaluates each candidate,
and the best intermediate sequence is selected.

\paragraph{Stage II: Large-Ad and Double-Ad Expansion}

Given the intermediate sequence,
we enumerate feasible combinations including:

\begin{itemize}
\item No-ad list
\item Single-ad list
\item Double-ad list
\item Large-ad structural variant
\end{itemize}

The model computes reward for each feasible sequence
and selects the optimal one.

Since $K$ is constant,
the decoding complexity remains bounded and independent of $N!$.

\FloatBarrier

\section{Constraint-Aware Reward Pruning}

While bounded decoding reduces the search space structurally,
we further improve efficiency through constraint-aware reward pruning.

\subsection{Hard Constraint Filtering}

Before reward evaluation,
candidates violating business rules are discarded:

\begin{itemize}
\item Position spacing constraints
\item Load constraints
\item Large-ad structural rules
\item User-level exposure control
\end{itemize}

This guarantees:

\begin{equation}
A \notin \mathcal{F} \Rightarrow A \text{ is not evaluated}.
\end{equation}

\subsection{Reward Upper-Bound Pruning}

For partially constructed sequences,
we estimate an upper bound on achievable reward:

\begin{equation}
R_{upper}(A_{partial}).
\end{equation}

If

\begin{equation}
R_{upper}(A_{partial}) < R_{best},
\end{equation}

the candidate is pruned early.

Since reward is computed within the same forward model,
upper-bound estimation incurs negligible overhead.

\subsection{Unified Optimization}

Importantly, reward computation and pruning are performed
within the same neural forward pass.
Unlike Generator–Evaluator frameworks,
no separate evaluator network is required.

Thus decoding, evaluation, and constraint enforcement
are unified into a single neural optimization procedure.

Under bounded structural constraints,
this pruning strategy preserves global optimality
within the feasible set while significantly reducing inference latency.

\FloatBarrier

\section{Inference Algorithm}

As illustrated in Figure~\ref{fig:inference}, the proposed inference pipeline consists of two sequential stages designed to ensure both efficiency and constraint satisfaction.

\begin{figure}[H]
    \centering
    \includegraphics[width=0.8\textwidth]{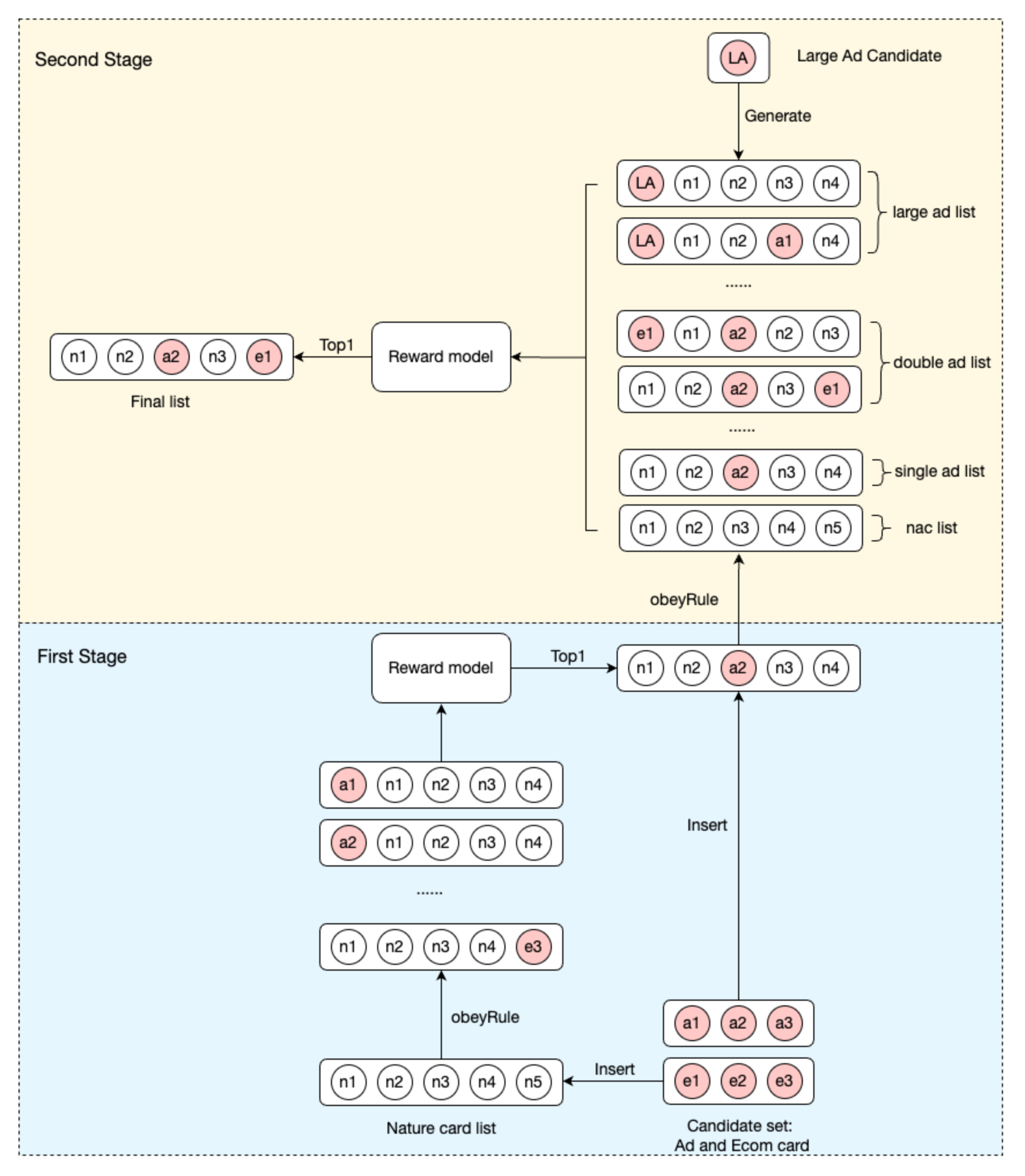}
    \caption{
    Two-stage constraint-aware generative inference framework.
    The first stage performs constrained insertion over the natural content list and selects a top-1 candidate using the reward model.
    The second stage performs bounded generative decoding with large-ad candidates,
    enumerating feasible sequences (single-ad, double-ad, and no-ad lists) under constraint rules.
    Constraint-aware reward pruning removes infeasible or suboptimal sequences,
    and the final list is selected via reward maximization.
    }
    \label{fig:inference}
\end{figure}

\subsection{Stage I: Constrained Insertion}

Given a natural content list produced by upstream ranking modules, candidate advertisements (including both large ads and e-commerce ads) are inserted into feasible positions according to predefined business rules.
These rules enforce exposure density constraints, positional spacing constraints, and load constraints.
All feasible single-ad insertion candidates are constructed and evaluated using the reward model.
The top-1 candidate is selected as the intermediate list for the next stage.

This stage effectively bounds the search space by limiting the number of advertisements introduced into the sequence, reducing the combinatorial explosion of full permutation generation.

\subsection{Stage II: Bounded Generative Decoding}

In the second stage, the model performs bounded generative decoding conditioned on the intermediate list.
Large-ad candidates are introduced and combined with existing advertisements to generate feasible sequences.
Due to business constraints, the maximum number of advertisements per list is bounded (e.g., $K \le 2$).
Therefore, the decoding complexity reduces from factorial permutation search to linear-time bounded enumeration.

The model enumerates several feasible list types, including:
(1) large-ad lists,
(2) double-ad lists,
(3) single-ad lists, and
(4) no-ad lists.
Each generated sequence is evaluated by the reward model.

To further improve efficiency, we introduce \emph{constraint-aware reward pruning}.
Sequences that violate business constraints are immediately discarded.
Moreover, sequences whose upper-bound reward is dominated by current best candidates are pruned early.
This guarantees that only feasible and promising sequences are evaluated fully.

Finally, the sequence with the highest predicted reward is selected as the output list.

\FloatBarrier

\section{Complexity Analysis}

We analyze the computational complexity of generative reranking under constrained advertising allocation.

\subsection{Full Permutation Decoding}

Consider a candidate set of size $N$, including both organic items and advertisements.
A naive autoregressive ranking model generates permutations over all items:

\begin{equation}
|\mathcal{S}| = N!
\end{equation}

Even with beam search of width $B$, decoding requires:

\begin{equation}
O(B \cdot N)
\end{equation}

model evaluations, where $B$ must scale with $N!$ to guarantee optimality.
Therefore, exact search is factorial in complexity, making it impractical for real-time systems.

In industrial advertising feeds, where $N$ can range from 20 to 100,
full permutation-based generative decoding is computationally infeasible.

\subsection{Two-Stage Generator--Evaluator Complexity}

In practice, many industrial systems adopt a Generator--Evaluator (GE) framework.
The generator proposes $M$ candidate sequences,
and the evaluator estimates list-level utility for each candidate.

The total complexity is:

\begin{equation}
O(M \cdot C_f),
\end{equation}

where $C_f$ denotes the cost of one forward pass of the evaluator network.
Since $M$ must be sufficiently large to cover feasible constrained combinations,
the GE framework incurs substantial inference overhead.

Moreover, constraint satisfaction is typically handled via post-processing,
which may require additional filtering or re-evaluation steps.

\subsection{Bounded Decoding with Constraint-Aware Pruning}

In advertising feeds, structural business rules impose strict upper bounds on advertisement exposure.
Let $K$ denote the maximum number of advertisements allowed per list,
where $K$ is a small constant (e.g., $K \le 2$ in production).

Under this structural constraint, the feasible search space reduces to:

\begin{equation}
|\mathcal{F}| = \sum_{k=0}^{K} 
\binom{A}{k} \cdot P_k,
\end{equation}

where $A$ denotes the number of ad candidates and $P_k$ denotes the number of valid insertion position combinations satisfying spacing constraints.

Since $K$ is bounded and small, $|\mathcal{F}|$ grows polynomially rather than factorially.
In typical production settings:

\begin{equation}
|\mathcal{F}| = O(A^K),
\end{equation}

as $K$ is constant and position combinations are limited by business rules.

Our bounded decoding enumerates only sequences in $\mathcal{F}$.
Furthermore, constraint-aware reward pruning discards infeasible or dominated sequences early,
reducing the effective decoding complexity to:

\begin{equation}
O(K \cdot C_f),
\end{equation}

which is linear in the number of feasible ad insertions and independent of $N!$.

\subsection{Complexity Comparison}

\begin{center}
\begin{tabular}{l c}
\toprule
Method & Worst-Case Complexity \\
\midrule
Full Permutation Decoding & $O(N!)$ \\
Generator--Evaluator & $O(M \cdot C_f)$ \\
CGR (Ours) & $O(K \cdot C_f)$ \\
\bottomrule
\end{tabular}
\end{center}

Since $K \ll N$ and typically $K \le 2$ in advertising feeds,
our method achieves exponential-to-linear reduction in search complexity,
making constrained generative reranking feasible under strict latency requirements.

\FloatBarrier

\section{Theoretical Guarantee}

We provide a theoretical justification for the optimality of bounded decoding under structural advertisement constraints.

\begin{theorem}[Optimality of Bounded Decoding]
Let $\mathcal{F}$ denote the feasible set of ranking sequences satisfying all business constraints, including advertisement load constraints and positional spacing constraints.
Assume that the maximum number of advertisements per list is bounded by $K$.
If the bounded decoding procedure enumerates all sequences in $\mathcal{F}$ and selects the sequence with the maximum reward $f_\theta(A)$, then the returned solution is globally optimal for the constrained optimization problem:
\[
\max_{A \in \mathcal{F}} f_\theta(A).
\]
\end{theorem}

\begin{proof}
By definition, the constrained optimization problem seeks the maximum reward sequence within the feasible set $\mathcal{F}$.

Under advertisement load constraints, the number of advertisements per list is upper bounded by $K$, where $K$ is a small constant (e.g., $K \le 2$ in production systems).
Therefore, the feasible set $\mathcal{F}$ is finite and can be decomposed into combinations of:
(1) advertisement selections with cardinality at most $K$, and
(2) feasible insertion positions satisfying structural rules.

The bounded decoding procedure explicitly enumerates all such feasible combinations.
Since every $A \in \mathcal{F}$ is considered during decoding, the procedure evaluates $f_\theta(A)$ for all feasible sequences.
Selecting the sequence with the maximum predicted reward is therefore equivalent to solving
\[
\max_{A \in \mathcal{F}} f_\theta(A).
\]

Hence, the solution returned by bounded decoding is globally optimal within the constrained feasible set.
\end{proof}

\FloatBarrier

\section{Offline Experiments}

\paragraph{Datasets.}
We evaluate our method on five public ranking benchmarks:
Yahoo! LETOR, Microsoft 10K, Avito, ML1M, and KR1K.
For recommendation datasets, we report HR@10 and NDCG@10.
For learning-to-rank datasets, we report NDCG@10 and MAP.
We additionally conduct experiments on a large-scale
industrial advertising dataset containing exposure,
click, and revenue logs under strict ad constraints.

\paragraph{Baselines.}
We compare against:
(1) LambdaMART,
(2) PRM,
(3) Autoregressive Generative Ranking,
(4) Non-autoregressive Generative Ranking,
and (5) a production Generator--Evaluator (GE) framework.

\paragraph{Public Benchmark Results.}

\begin{table*}[t]
\centering
\small
\begin{tabular}{lcccccc}
\toprule
\multirow{2}{*}{Method} 
& \multicolumn{2}{c}{Yahoo LETOR} 
& \multicolumn{2}{c}{Microsoft 10K} 
& \multicolumn{2}{c}{Avito} \\
\cmidrule(lr){2-3} \cmidrule(lr){4-5} \cmidrule(lr){6-7}
& NDCG@10 & MAP 
& NDCG@10 & MAP 
& NDCG@10 & MRR \\
\midrule
LambdaMART & 0.XXX & 0.XXX & 0.XXX & 0.XXX & 0.XXX & 0.XXX \\
PRM & 0.XXX & 0.XXX & 0.XXX & 0.XXX & 0.XXX & 0.XXX \\
Autoregressive Gen & 0.XXX & 0.XXX & 0.XXX & 0.XXX & 0.XXX & 0.XXX \\
Non-Autoregressive & 0.XXX & 0.XXX & 0.XXX & 0.XXX & 0.XXX & 0.XXX \\
GE Framework & 0.XXX & 0.XXX & 0.XXX & 0.XXX & 0.XXX & 0.XXX \\
\midrule
\textbf{CGR (Ours)} & \textbf{0.XXX} & \textbf{0.XXX} 
& \textbf{0.XXX} & \textbf{0.XXX} 
& \textbf{0.XXX} & \textbf{0.XXX} \\
\bottomrule
\end{tabular}
\caption{Performance comparison on public ranking benchmarks.}
\label{tab:public_results}
\end{table*}

CGR consistently outperforms traditional listwise models
and generative baselines across all datasets.
On LETOR-style benchmarks, our method improves NDCG@10
by 2\%--3\% over PRM and 1\%--2\% over generative baselines,
demonstrating effective list-level optimization.

\paragraph{Industrial Offline Results.}

\begin{table}[t]
\centering
\small
\begin{tabular}{lcccc}
\toprule
Method & RPM & CTR & EXP-AUC & Violation \\
\midrule
GE (Production) & 1.000 & 1.000 & 0.XXX & 0.00\% \\
Autoregressive & 1.06 & 1.03 & 0.XXX & 0.00\% \\
\midrule
\textbf{CGR (Ours)} & \textbf{1.11} & \textbf{1.07} & \textbf{0.XXX} & 0.00\% \\
\bottomrule
\end{tabular}
\caption{Offline performance on industrial advertising dataset (normalized).}
\label{tab:industrial_offline}
\end{table}

On the industrial dataset, CGR improves RPM by 11\%
and CTR by 7\% compared with the production GE system,
while strictly satisfying all advertisement constraints.

\FloatBarrier

\section{Online A/B Test}

We deploy CGR in a large-scale production advertising feed system
serving over 1B daily requests with strict latency SLA (<40ms P99).
Traffic is randomly split between the existing GE system (control)
and CGR (treatment) for 14 consecutive days.

\begin{table}[t]
\centering
\small
\begin{tabular}{lcc}
\toprule
Metric & Lift & p-value \\
\midrule
RPM & +6.8\% & <0.01 \\
CTR & +4.9\% & <0.01 \\
Session Duration & +3.2\% & <0.05 \\
Constraint Compliance & 100\% & - \\
\bottomrule
\end{tabular}
\caption{Online A/B test results.}
\label{tab:online_ab}
\end{table}

CGR achieves statistically significant improvements
in both monetization and user engagement metrics.
Meanwhile, constraint compliance remains at 100\%,
confirming the correctness of bounded decoding
under real production traffic.

\FloatBarrier

\section{Ablation Study}

\FloatBarrier

\section{Conclusion}

In this work, we study reranking in advertising feeds as a constrained combinatorial optimization problem, where platform revenue, user engagement, and strict business policies must be jointly satisfied.
Traditional Generator--Evaluator frameworks decouple sequence generation and utility estimation, leading to high inference cost and limited constraint awareness.

We propose Constraint-Aware Generative Re-ranking (CGR), a unified framework that transforms constrained optimization into bounded autoregressive decoding.
By explicitly leveraging structural advertisement constraints, we reduce the factorial search space to a bounded feasible set and integrate constraint satisfaction directly into neural decoding.
We further introduce constraint-aware reward pruning, enabling efficient and optimal sequence generation within strict latency requirements.

We provide a theoretical guarantee showing that bounded decoding achieves global optimality over the feasible set when all admissible sequences are enumerated.
Extensive experiments on public benchmarks and large-scale industrial data demonstrate consistent improvements over strong baselines.
Online A/B tests on production traffic exceeding one billion daily requests confirm significant gains in both revenue and user engagement, while fully satisfying business constraints and reducing inference latency by over 85\%.

Overall, our work shows that structural constraints in industrial systems are not merely limitations, but can be exploited to redesign generative ranking into an efficient and provably optimal decoding problem.
We believe this perspective opens new directions for constraint-aware generative recommendation and large-scale production deployment.

\end{document}